\documentclass{article}
\usepackage{spconf,amsmath,graphicx}
\usepackage{url,cite} 
\usepackage{amssymb}
\usepackage{subfigure}
\usepackage{enumitem}
\usepackage{blindtext}
\usepackage{graphicx}
\usepackage{amsmath}
\usepackage{amsfonts}
\usepackage{braket}
\usepackage{multirow}
\usepackage{cite}
\usepackage[ruled]{algorithm2e}
\usepackage[utf8]{inputenc}
\usepackage[english]{babel}
\usepackage{amsthm}

\setitemize{noitemsep,topsep=3pt,parsep=0.5pt,partopsep=0.5pt}
\setenumerate{noitemsep,topsep=3pt,parsep=0.5pt,partopsep=0.5pt}

\title{Interference-Precancelled Pilot Design for LMMSE Channel Estimation of GFDM}

\name{Ching-Lun Tai$^1$, Borching Su$^2$, Cai Jia$^2$}
\address{$^1$Department of Electrical Engineering, National Taiwan University, Taiwan\\
$^2$Graduate Institute of Communication Engineering, National Taiwan University, Taiwan}

\newcommand{\vd}{{\bf d}}

\newcommand{\vg}{{\bf g}}
\newcommand{\vh}{{\bf h}}

\newcommand{\vp}{{\bf p}}
\newcommand{\vq}{{\bf q}}

\newcommand{\vu}{{\bf u}}

\newcommand{\vw}{{\bf w}}
\newcommand{\vx}{{\bf x}}
\newcommand{\vy}{{\bf y}}

\newcommand{\zv}{{\bf 0}}

\newcommand{\cA}{{\cal A}}

\newcommand{\cI}{{\cal I}}

\newcommand{\cK}{{\cal K}}

\newcommand{\cM}{{\cal M}}

\newcommand{\bC}{{\mathbb{C}}}

\newcommand{\bN}{{\mathbb{N}}}

\newenvironment{mat}[1]{\left[\begin{array}{#1}}{\end{array}\right]}

\newcommand{\mx}[1]{\begin{mat}{ccccccc}#1\end{mat}}

\renewcommand{\set}[1]{\left\{#1\right\}}

\newcommand{\mA}{{\bf A}}

\newcommand{\mF}{{\bf F}}
\newcommand{\mG}{{\bf G}}
\newcommand{\mH}{{\bf H}}
\newcommand{\mI}{{\bf I}}

\newcommand{\mP}{{\bf P}}

\newcommand{\mS}{{\bf S}}
\newcommand{\mT}{{\bf T}}

\newcommand{\mW}{{\bf W}}
\newcommand{\mX}{{\bf X}}

\newcommand{\ignore}[1]{}

\DeclareMathOperator{\tr}{tr}
\DeclareMathOperator{\e}{E}

\DeclareMathOperator{\diag}{diag}

\newcommand{\modd}[1]{\left\langle#1\right\rangle}
\newtheorem{theorem}{Theorem}

\begin{document}

\maketitle
\begin{abstract}
Generalized frequency division multiplexing (GFDM) is a promising candidate waveform for next-generation wireless communication systems. However, GFDM channel estimation is still challenging due to the inherent interference. In this paper, we formulate a pilot design framework with linear minimum mean square error (LMMSE) channel estimation for GFDM, and propose a novel pilot design to achieve interference precancellation during pilot generation with the fixed transmit sample values at selected frequency bins. Numerical results demonstrate that the proposed method reduces the channel estimation mean square error and the symbol error rate (SER) in high signal-to-noise ratio (SNR) regions, compared with the conventional methods.
\end{abstract}
\begin{keywords}
Generalized frequency division multiplexing (GFDM),
linear minimum mean square error (LMMSE) channel estimation, pilot design, symbol error rate (SER)
\end{keywords}

\section{Introduction}
Generalized frequency division multiplexing (GFDM), considered as a candidate waveform for next-generation wireless communication systems, features several advantages such as low out-of-band (OOB) emissions and relaxed requirements of time and frequency synchronizations \cite{michailow14}.
However, unlike traditional OFDM, which is free from interference, GFDM suffers from the inherent inter-subsymbol interference (ISI) and potential inter-subcarrier interference (ICI), and thus the channel estimation is crucial but rather challenging for GFDM.

To further improve the channel estimation performance, several methods have been proposed.
In \cite{ehsanfar17,na18}, the waveform and structure of GFDM have been modified for better accuracy of channel estimation, resulting in the prototype filters for data symbols in pilot subcarriers not satisfying the orthogonality condition. 
Moreover, the technique of orthogonal match pursuit (OMP) is adopted in \cite{zhang16b}. However, the iterative method proposed in \cite{zhang16b} might cause a severe latency to the system, which is against the requirements of next-generation wireless communication systems.

In addition to the above methods, the linear channel estimation has been widely studied. The matched filter (MF) has been employed in \cite{vilaipornsawai14,danneberg15a}, but this approach is based on the assumption of nearly flat and slow fading channels, which are unrealistic in broadband communication. On the other hand, the least square (LS) and linear minimum mean square error (LMMSE) estimation have been adopted in \cite{ehsanfar16,ehsanfar16b,akai17}. Though these methods do not require any additional assumption, they suffer from estimation performance degradation due to the inherent interference in GFDM.


Despite the importance of pilot design for channel estimation, it is less studied in the GFDM scenario. In \cite{tang17}, the concept of in-phase and quadrature (IQ) imbalance compensation is adopted for pilot generation. Furthermore, the carrier frequency offset (CFO) estimation is considered for pilot design in \cite{shayanfar18}. However, these schemes assume a perfect interference cancellation, which is rather unrealistic.

In this paper, we formulate a pilot design framework and its corresponding LMMSE channel estimator for GFDM. 
With the framework, we propose a novel pilot design for interference precancellation during pilot generation to improve GFDM's LMMSE channel estimation accuracy, which is much closer to that of OFDM, compared with previous methods \cite{ehsanfar16,ehsanfar16b,akai17}.

The remainder of this paper is organized as follows. Section \ref{sec:system} describes the GFDM system model. In Section \ref{sec:method}, we first introduce the mathematical expressions of pilot design and corresponding LMMSE channel estimation, and then present the proposed method. Simulation results and discussions are provided in Section \ref{sec:simulation}. Finally, Section \ref{sec:conclusion} concludes the paper.

\emph{Notations:}
Boldfaced capital and lowercase letters denote matrices and column vectors, respectively. We use $\e\{\cdot\}$ and $\modd{.}_D$ to denote the expectation operator and modulo $D$, respectively.
For any set $\cA$, we use $|\cA|$ to denote its cardinality.
We adopt the MATLAB subscripts $:$ and $a:b$ to denote all elements and the elements ordered from $a$ to $b$, respectively, of the subscripted objects.
Given a vector $\vu$, we use $[\vu]_n$ to denote the $n$th component of $\vu$, $\|\vu\|$ the $\ell_2$-norm of $\vu$, and $\diag(\vu)$ the diagonal matrix containing $\vu$ on its diagonal.  
Given a matrix $\mA$, we denote $[\mA]_{m, n}$, $\tr(\mA)$, ${\mA}^T$, ${\mA}^*$, and ${\mA}^H$ its ($m$, $n$)th entry (zero-based indexing), trace, transpose, complex conjugate, and Hermitian transpose, respectively. 
We define ${\mI}_q$ to be the $q \times q$ identity matrix, 
$\zv_q$ the $q \times 1$ zero vector, and ${\mW}_q$ the normalized $q$-point discrete Fourier transform (DFT) matrix with ${[\mW_q]}_{m, n}=e^{-j2\pi mn/q}/\sqrt{q}, q \in \bN$.
\section{GFDM System Model}
\label{sec:system}
\begin{figure}[h]
\begin{center}
\includegraphics[width=8.5cm]{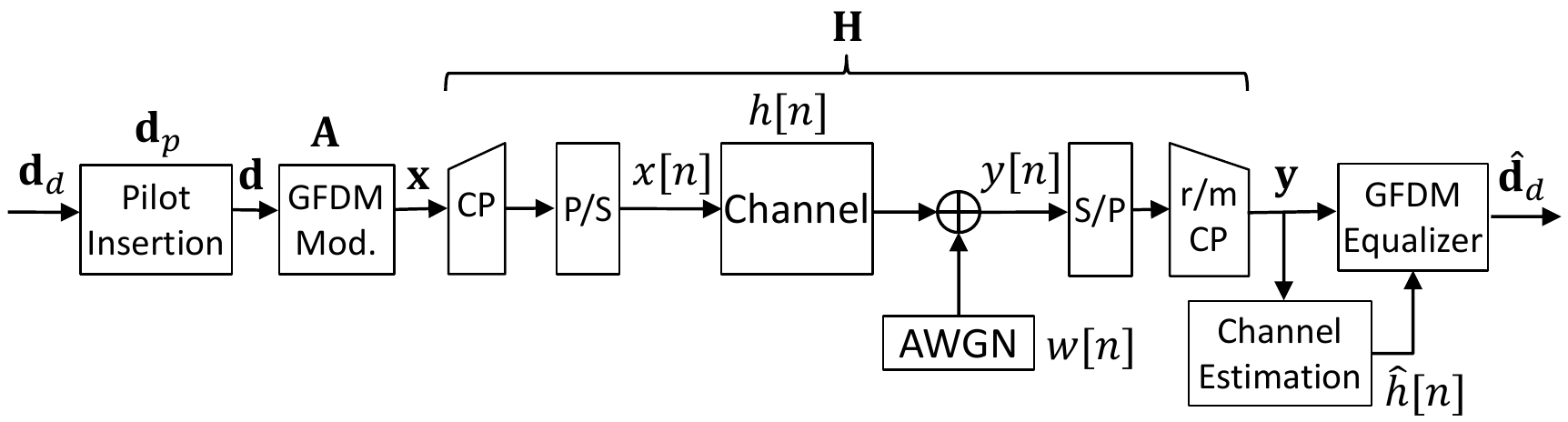}
\caption{GFDM system model with channel estimation (``r/m'' stands for ``remove''.)}
\label{fig:system}
\end{center}
\end{figure}
GFDM is a block-based communication scheme as shown in Fig. \ref{fig:system} \cite{michailow14, chen17b}. Each GFDM block employs $K$ subcarriers, with each transmitting $M$ complex-valued subsymbols. So, a total of $D=KM$ symbols are transmitted in a block. Let $\vd_l \in \bC^D$ be the $l$th GFDM block, whose $m$th subsymbol on the $k$th subcarrier is denoted as $[\vd_l]_{k+mK}$, $m=0,1,...,M-1, k=0,1,...,K-1.$ 

Each symbol $[\vd_l]_{k+mK}$ is pulse-shaped by a vector $\vg_{k,m}$, whose $n$th entry is $[\vg_{k,m}]_n=[\vg]_{\modd{n-mK}_D}e^{j2 \pi kn/K}$, $n=0, 1, ..., D-1$, $m=0,1,...,M-1,k=0,1,...,K-1$, where $\vg \in \bC^D$ is called the \emph{prototype filter} \cite{michailow14}. Let

\begin{equation}
\mA=[\vg_{0,0}...\vg_{K-1,0} \quad \vg_{0,1}... \vg_{K-1,1}...\vg_{K-1,M-1}]
\label{eq:A}
\end{equation}
be the GFDM transmitter matrix \cite{michailow14} and $\vx_l=\mA \vd_l$ be the transmit sample vector, whose $n$th entry, for $n=0, 1, ..., D-1$, is
\begin{equation}
{[\vx_l]}_n = \sum_{k=0}^{K-1} \sum_{m=0}^{M-1} [\vd_l]_{k+mK}[\vg]_{\modd{n-mK}_D}e^{j2 \pi kn/K}.
\label{eq:xln}
\end{equation}

Subsequently, the vector $\vx_l$ is passed through parallel-to-serial (P/S) conversion, and a cyclic prefix (CP) of length $L$ is further added. Denote the set of subcarrier indices and set of subsymbol indices that are actually employed as $\cK\subseteq \set{0,1,...,K-1}$ and $\cM\subseteq\set{0,1,...,M-1}$, respectively. The digital baseband transmit signal of GFDM can be expressed as \cite{chen17b}
{\fontsize{8}{6}
\begin{equation}
x[n] = \sum_{l=-\infty}^\infty \sum_{k \in \cK} \sum_{m \in \cM} [\vd_l]_{k+mK}g_m[n-lD']e^{j2 \pi k(n-lD')/K} \label{eq:xn},
\end{equation}
}
where $D'=D+L$ and 
\begin{equation}
g_m[n] = \left\{\begin{array}{ll}
[\vg]_{\modd{n-mK-L}_D}, & n = 0, 1, ..., D'-1\\
0, & \mbox{otherwise}
\end{array}\right..
\label{eq:gmn}
\end{equation}
For notational brevity, we omit the subscript $l$ as in $\vx_l$ and $\vd_l$ hereafter. 

As shown in Fig. \ref{fig:system}, the received signal after transmission through a wireless channel can be modeled as a linear time-invariant (LTI) system $y[n]=h[n]*x[n]+w[n]$, where $h[n]$ is the channel impulse response, and $w[n]$ is the complex additive white Gaussian noise (AWGN) with variance $N_0$.
We denote $\vw=[w[0] w[1]...w[D-1]]^T$ and $\vh=[h[0] h[1]...h[N-1]]^T$, where $N-1$ is the channel order, as the vector forms of complex AWGN and channel impulse response, respectively.
Note that $\vh=\sqrt{\mbox{diag}(\vp)}\vq \in \bC^N$, where $\vp\in\bC^N$ is the power delay profile (PDP) and $\vq\in\bC^N$ is a vector of independently and identically distributed (i.i.d.) standard normal random variables, with the assumption that the channel order $N-1$ does not exceed the CP length $L$.
After CP removal and serial-to-parallel (S/P) conversion, the received samples are \cite{michailow14}
{\fontsize{9}{6}
\begin{equation}
\vy=\mH \vx+\vw=\mH \mA \vd+\vw=\mW_D^H\mbox{diag}(\mW_D\mA\vd)\mF_N\vh+\vw,
\label{eq:y}
\end{equation}
}where $\mF_N=[\sqrt{D}\mW_D]_{:,1:N}$, and $\mH\in\bC^{D \times D}$ is the circulant matrix whose first column is $[\vh^T \quad \zv_{D-N}^T]^T$.

To obtain the channel state information (CSI), several spaces in the GFDM block $\vd$ are reserved for the pilot vector $\vd_p$ through pilot insertion, and the rest of spaces in $\vd$ are employed to transmit the data vector $\vd_d$.
The channel estimation is performed upon the arrival of the received samples $\vy$, leading to the reconstructed data vector $\hat{\vd}_d$, as shown in Fig. \ref{fig:system}.
\section{Proposed Method}
\label{sec:method}
In this section, the issue of pilot design for LMMSE channel estimation of GFDM systems is considered. Due to the inherent interference in GFDM, the performance of LMMSE channel estimation in existing literature \cite{ehsanfar16,ehsanfar16b,akai17} suffers from severe degradation, especially at high signal-to-noise ratio (SNR), due to the lack of interference precancellation in pilot design. To tackle this problem, we will first introduce the mathematical formulation regarding pilot design and LMMSE channel estimation, and then illustrate the proposed method.
\begin{figure}[h]
    \centering
    \includegraphics[width=8.5cm]{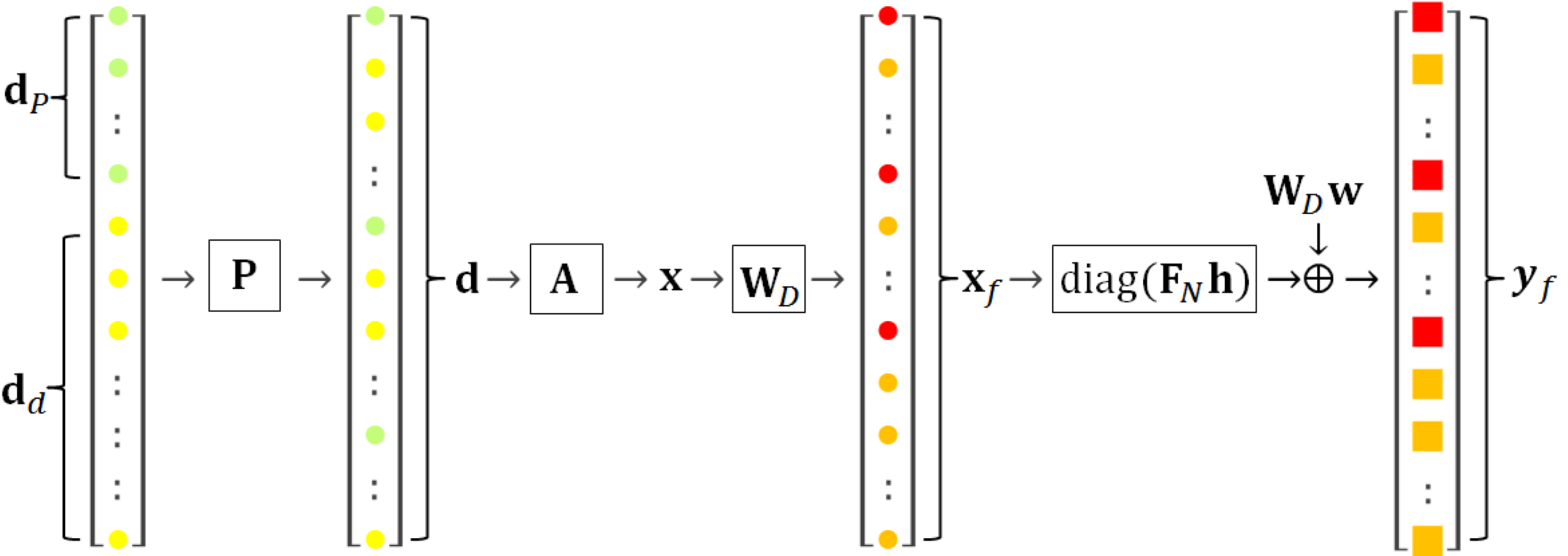}
    \caption{Pilot insertion and frequency bin deployment, with green solid circles denoting the elements in $\vd_p$, yellow solid circles the elements in $\vd_d$, red solid circles and squares the chosen frequency bins of $\vx_f$ and $\vy_f$, orange circles and squares the rest of frequency bins of $\vx_f$ and $\vy_f$, where $\vy_f=\diag(\mF_N\vh)\vx_f+\mW_D\vw$, and the relationship between $\vx_f$ and $\vy_f$ benefiting the frequency-domain channel estimation}
    \label{fig:pilot}
\end{figure}
\vspace{-0.15cm}
\subsection{Pilot Design and LMMSE Channel Estimation}
\vspace{-0.1cm}
In order to achieve the channel estimation at the receiver, the GFDM block $\vd$ is generated from two subvectors, the pilot vector $\vd_p\in\bC^p$ and the data vector $\vd_d\in\bC^{D-p}$, as shown in Fig. \ref{fig:pilot}, where $p\in\bN$ denotes the number of pilots in a block.

Accordingly, the GFDM block $\vd$ can be expressed as
\begin{equation}
\vd=\underbrace{[\mP_1 \quad \mP_2]}_{\mP}\mx{\vd_p\\ \vd_d},
\label{eq:dpdd}
\end{equation}
where $\mP$ is a $D \times D$ permutation matrix, and $\mP_1\in\bC^{D \times p}$ and $\mP_2\in\bC^{D \times (D-p)}$ are submatrices of $\mP$.
In addition, given a reference sequence $\vd_r\in\bC^p$, the GFDM block $\vd$ can be represented as a linear transformation of $\vd_r$ and $\vd_d$, i.e., 
\begin{equation}
\vd=\mS \vd_r+\mT \vd_d,
\label{eq:data}
\end{equation}
where $\mS\in\bC^{D \times p}$ and $\mT\in\bC^{D \times (D-p)}$ are the corresponding linear coefficient matrices of $\vd_r$ and $\vd_d$, respectively.
Assume that $\vd_r$ and $\vd_d$ are independent, and the symbols in $\vd_d$ are zero-mean and i.i.d. with symbol energy $E_s$, i.e., $\e\{\vd \vd^H\}=\mS\e\{\vd_r\vd_r^H\}\mS^H+E_S\mT\mT^H.$
With different choices of $\mS$ and $\mT$, the pilot vector $\vd_p$ in (\ref{eq:dpdd}) is accordingly modified. Based on this pilot design framework, different pilots can be generated in order to meet specific requirements.
In the following theorem, we derive the LMMSE channel estimation corresponding to the pilot design framework.
\begin{theorem}
\label{Thm1}
Given the received samples $\vy$ as defined in (\ref{eq:y}), the LMMSE estimated channel $\hat{\vh}_{\footnotesize{\mbox{LMMSE}}}$ can be derived as
\begin{equation}
\hat{\vh}_{\footnotesize{\mbox{LMMSE}}}=\mG_{\footnotesize{\mbox{LMMSE}}}\vy,
\label{eq:LMMSE}
\end{equation}
with
{\fontsize{8}{6}
\begin{align}
\mG_{\footnotesize{\mbox{LMMSE}}}=&\mathbf{\Sigma}_{hh}(\mX_r \mF_N)^H\nonumber\\
\times&[(\mX_r \mF_N) \mathbf{\Sigma}_{hh} (\mX_r \mF_N)^H+\mathbf{\Sigma}_{\Psi \Psi}+N_0\mI_D]^{-1}\mW_D,
\end{align}
}where $\mX_r=\mbox{diag}(\mW_D\mA\mS\vd_r)$,  $\mathbf{\Sigma}_{hh}=\e\{\vh\vh^H\}=\mbox{diag}(\vp)$, and $\mathbf{\Sigma}_{\Psi \Psi}=(\mF_N \mathbf{\Sigma}_{hh} \mF_N^H)(\mW_D \mA \mT \mT^H \mA^H \mW_D^H)$.
\end{theorem}
\begin{proof}
Let the channel estimator be $\mG$, and $\hat{\vh}=\mG\vy$ is the estimated channel.
Note that $\diag(\mW_D\mA\vd)=\mX_r+\mX_d$.
First, we derive the expected square error of channel estimation $\e\{\|\vh-\mG\vy\|^2\}$ as
\begin{align}
      & \e\{\tr((\vh-\mG\vy)(\vh-\mG\vy)^H)\}\nonumber\\
    =& \tr(\mathbf{\Sigma}_{hh}-\e\{\vh\vy^H\mG^H\}-\e\{\mG\vy\vh^H\}+\e\{\mG\vy\vy^H\mG^H\})\nonumber\\
    =& \tr(\mathbf{\Sigma}_{hh})-\tr(\mathbf{\Sigma}_{hh}(\mW_D^H\mX_r\mF_N)^H\mG^H)\nonumber\\
    -& \tr(\mG(\mW_D^H\mX_r\mF_N)\mathbf{\Sigma}_{hh})+\tr(\mG[\mW_D^H(\mX_r\mF_N)\mathbf{\Sigma}_{hh}\nonumber\\
    \times& (\mX_r\mF_N)^H\mW_D+\mW_D^H\mathbf{\Sigma}_{\Psi\Psi}\mW_D+N_0\mI_D]\mG^H). 
    \label{expect_error}
\end{align}
The Wirtinger derivatives of $\e\{\|\vh-\mG\vy\|^2\}$ with regard to $\mG^*$ are obtained as \cite{gesbert07}
{\fontsize{7}{6}
\begin{align}
    \label{diff}
    \frac{\partial \e\{\|\vh-\mG\vy\|^2\}}{\partial \mG^*}= & -\mathbf{\Sigma}_{hh}(\mX_r\mF_N)^H\mW_D+ \mG[\mW_D^H(\mX_r\mF_N)\mathbf{\Sigma}_{hh}\nonumber\\
    \times&(\mX_r\mF_N)^H\mW_D+\mW_D^H\mathbf{\Sigma}_{\Psi\Psi}\mW_D+N_0\mI_D].
\end{align}
}To solve for the optimal estimator, we require the derivatives to be zero and derive that
{\fontsize{7}{6}
\begin{align}
\mG_{\footnotesize{\mbox{LMMSE}}}=&\mathbf{\Sigma}_{hh}(\mX_r \mF_N)^H\nonumber\\
\times&[(\mX_r \mF_N) \mathbf{\Sigma}_{hh} (\mX_r \mF_N)^H+\mathbf{\Sigma}_{\Psi \Psi}+N_0\mI_D]^{-1}\mW_D.
\end{align}
}
\end{proof}
A special case of Theorem \ref{Thm1} has been found in the pilot design of the conventional LMMSE channel estimation methods \cite{ehsanfar16,ehsanfar16b,akai17}, where the matrices $\mS$ and $\mT$ are set to satisfy $[\mS \quad \mT]=\mI_D$. In this case, this implies $\vd_p=\vd_r$ and $\mP=\mI_D$.
\vspace{-0.15cm}
\subsection{Proposed Choice of Parameters}
\vspace{-0.1cm}
In order to reduce the impact of interference on the LMMSE channel estimation for GFDM, the proposed method precancels the interference during pilot generation. We will first pick up the appropriate frequency bins indexed as $\cI\subset\{1,2,...,D\}$ (illustrated as red circles and squares shown in Fig. \ref{fig:pilot}, and require that
\begin{equation}
[\vx_f]_{\cI}=\vd_r,
\label{eq:xfI}
\end{equation}
where $\vx_f=\mW_D\vx$ is the frequency-domain transmit samples. Then, the pilots will be generated according to the requirement, leading to the effect of interference precancellation in pilot design.
Note that the channel can be exactly reconstructed in the noiseless case only if both $p \geq N$ and $|\cI| \geq p$ are satisfied.

To illustrate the proposed method, we will derive the corresponding $\mS$ and $\mT$ as follows.
Define the two matrices $\mW_1=[\mW_D \mA \mP]_{\cI,1:p}$ and $\mW_2=[\mW_D \mA \mP]_{\cI,(p+1):D}$. 
Based on the expressions (\ref{eq:dpdd}) and (\ref{eq:xfI}), we require that $\vd_r=\mW_1\vd_p+\mW_2\vd_d$,
and it can be derived that
\begin{equation}
\vd_p=\mW_1^{-1}(\vd_r-\mW_2\vd_d).
\label{eq:dp}
\end{equation}
From (\ref{eq:dp}), it can be observed that the pilot generation in the proposed method exploits the information about $\mP$, $\vd_r$, and $\vd_d$. The information about $\mP$ and $\vd_r$ is fundamentally available at the transmitter, and the information about $\vd_d$ can be fed into the transmitter before transmission.

Substituting (\ref{eq:dp}) into $\vd_p$ in (\ref{eq:dpdd}), we obtain
\begin{equation}
    \vd=\mP_1\mW_1^{-1}(\vd_r-\mW_2\vd_d)+\mP_2\vd_d.
\end{equation}Finally, the corresponding $\mS$ and $\mT$ of the proposed method can be expressed as
\begin{equation}
\label{eq:ST}
(\mS,\mT)=(\mP_1\mW_1^{-1},\mP_2-\mP_1\mW_1^{-1}\mW_2).
\end{equation}
Note that the expressions (\ref{eq:dp})-(\ref{eq:ST}) exist if and only if the index set $\cI$ mentioned above is chosen such that the matrix $\mW_1$ is invertible.
In the next section, we will choose $\cI$ such that the above condition satisfies.
\section{Simulation}
\label{sec:simulation}
In this section, numerical results are presented to compare the performances, in terms of mean square error (MSE) of LMMSE channel estimation and symbol error rate (SER), of the proposed method with those of the conventional methods \cite{ehsanfar16,ehsanfar16b,akai17} for two cases. Both the Dirichlet \cite{matthe14a} and raised cosine (RC) \cite{michailow14} filters are employed for GFDM channel estimation. In addition, OFDM channel estimation is included in the experiments for a comprehensive comparison.
\vspace{-0.15cm}
\subsection{Parameter Settings}
\vspace{-0.1cm}
The modulation is QPSK, the symbol energy is $E_s=1$, the CP length is $L=16$, and the roll-off factor of the RC filter is $\alpha=0.9$.
Consider two cases $(K,M)=(8,128)$ and $(16,64)$ for GFDM.
For a fair comparison, the same block size is used for OFDM.
For each case, the channel length is $N=K$ (with channel order $K-1$), the pilot vector length is $p=K$, and the set of frequency bin indexes is $\cI=\{1,M+1,2M+1,...,(K-1)M+1\}$ with $|\cI|=p$.
We set $\cK=\{0,1,...,K-1\}$ and $\cM=\{0,1,...,M-1\}$, i.e., all subcarriers and subsymbols are used.
To evaluate the performances, Monte Carlo simulation is adopted with randomly generated channel realizations and independent data sets for the realizations. We generate $N_h=100$ spatially uncorrelated Rayleigh-fading multipath channel realizations, whose channel PDP $\vp$ is exponential from 0 to -10 dB with $N$ taps, and $N_d=100$ independent data blocks for each channel realization.
Moreover, the genie-aided condition, where full CSI is known, is included and considered as the performance bound for the SER evaluation of all schemes.
\vspace{-0.15cm}
\subsection{Simulation Results}
\vspace{-0.1cm}
For the case $(K,M)=(8,128)$, the simulation results are shown in Fig. \ref{fig:K8M128_all}, where we compare the proposed method with conventional methods and OFDM in terms of MSE of LMMSE channel estimation and SER, respectively.
In Fig. \ref{fig:K8M128_MSE}, the MSE is calculated by averaging through the squares of pairwise Euclidean norms between a channel realization and its LMMSE estimated counterpart. According to Fig. \ref{fig:K8M128_MSE}, the proposed method significantly outperforms conventional methods, especially at high SNR where the impact of interference is much larger than that of noise, since the proposed method precancels the effect of interference during pilot generation. 

When we interpret the SER performance presented in Fig. \ref{K8M128_SER}, it can be observed that the proposed method performs comparably to conventional methods at low SNR, but outperforms conventional methods at high SNR. This is because the process of interference precancellation will increase the energy of pilot vector $\vd_p$, which degrades the SER performance.
Note that the Dirichlet filter outperforms the RC filter, whose transmitter matrix is not unitary when both $K$ and $M$ are even and leads to the SER performance degradation.
In addition, OFDM still performs the best among all schemes, with its performance gap between the curves of LMMSE channel estimation and genie-aided scheme smaller than that of GFDM.

For the case $(K,M)=(16,64)$, the simulation results demonstrate similar trends. Note that compared with the previous case, the MSE performance slightly deteriorates, while the SER performance slightly improves. 
\begin{figure}%
\centering
\subfigure[Channel MSE]{%
\label{fig:K8M128_MSE}%
\includegraphics[height=1.5in]{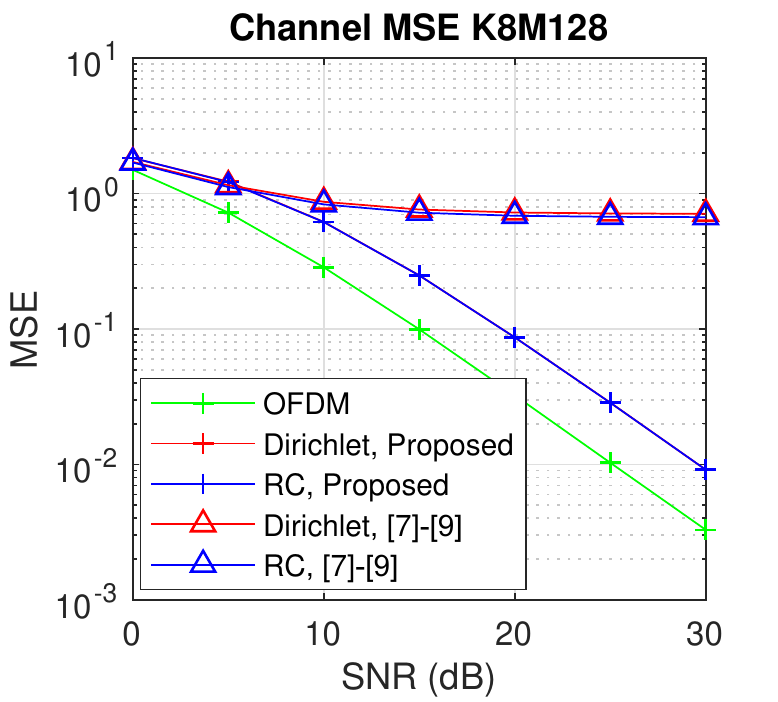}}%
\hspace{-0.4cm}
\subfigure[SER]{%
\label{K8M128_SER}%
\includegraphics[height=1.5in]{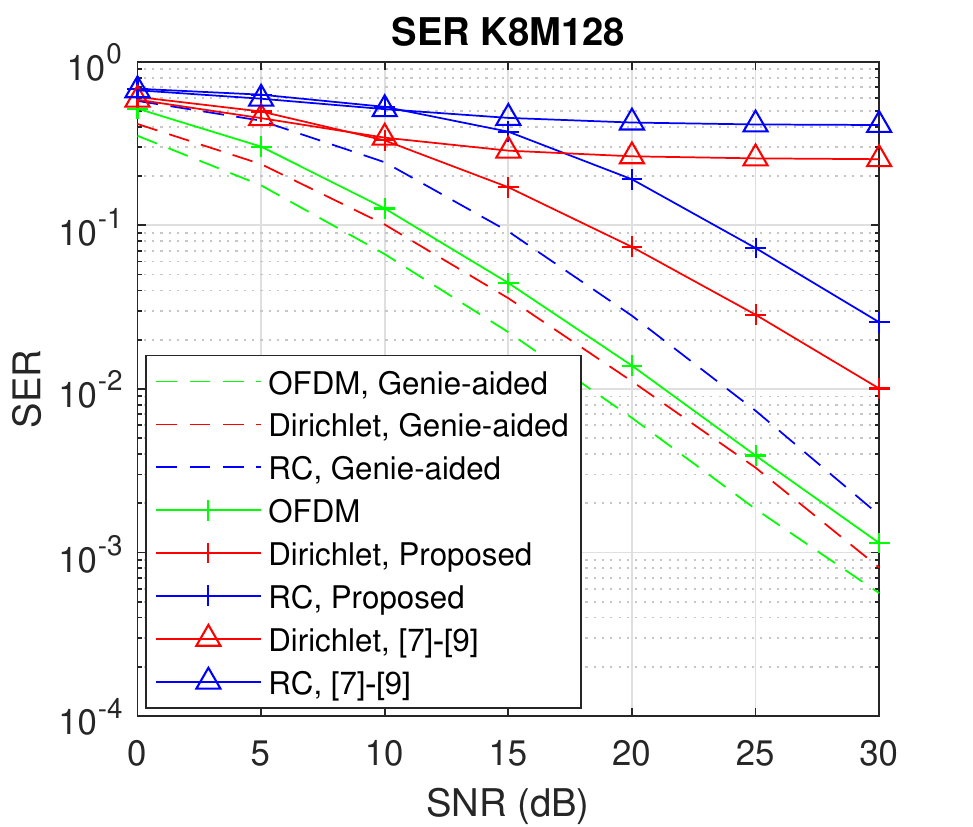}}%
\caption{Performance comparison for $K=8,M=128$}
\label{fig:K8M128_all}
\end{figure}

\begin{figure}%
\centering
\subfigure[Channel MSE]{%
\label{fig:K16M64_MSE}%
\includegraphics[height=1.5in]{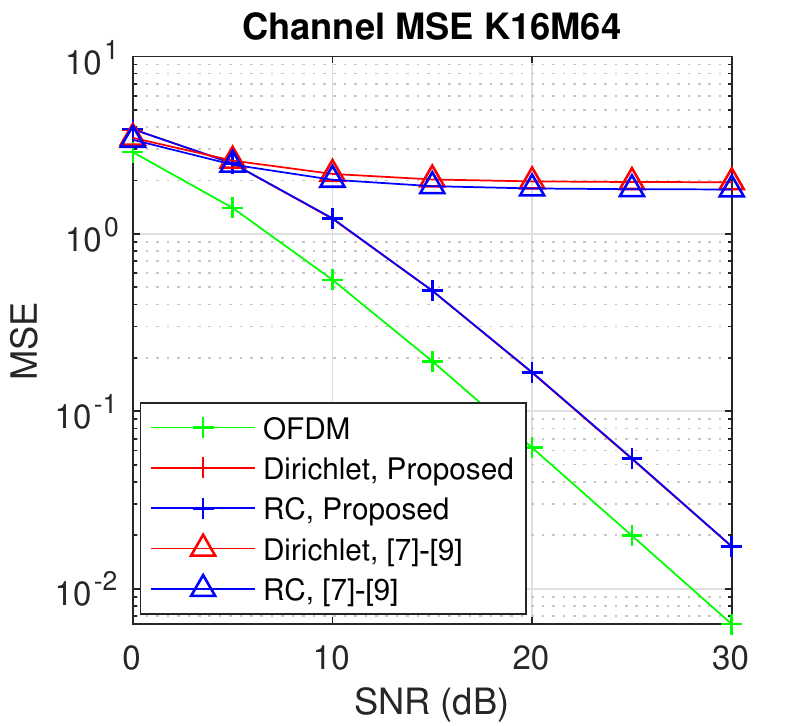}}%
\hspace{-0.4cm}
\subfigure[SER]{%
\label{K16M64_SER}%
\includegraphics[height=1.5in]{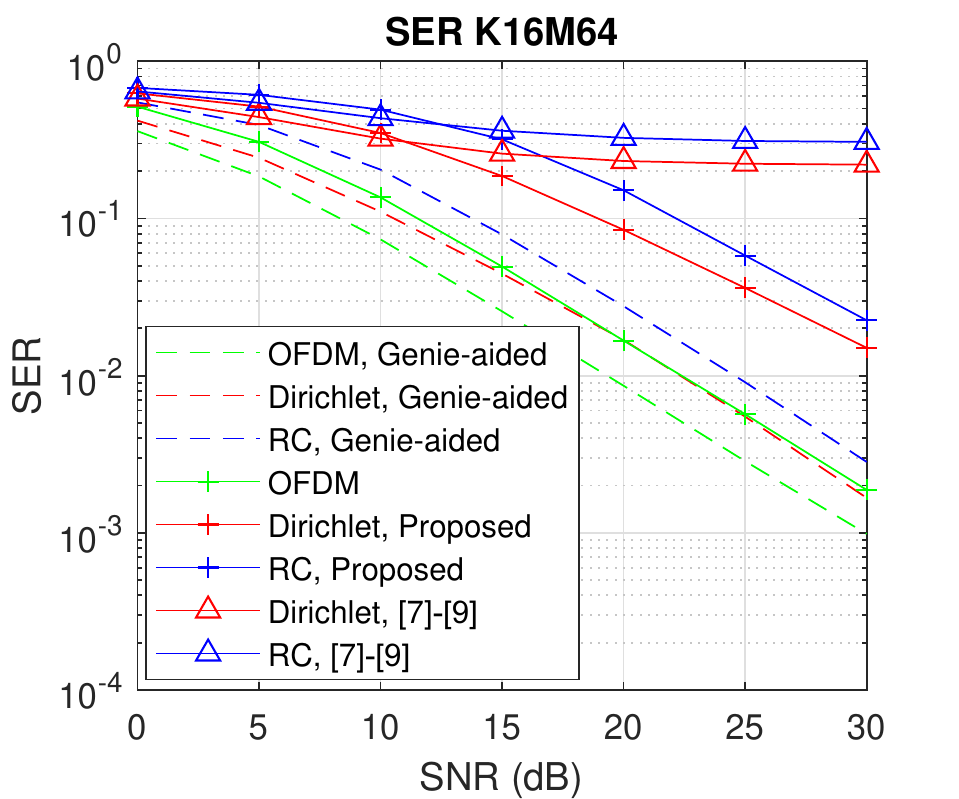}}%
\caption{Performance comparison for $K=16,M=64$}
\label{fig:K16M64_all}
\end{figure}

\section{Conclusion}
\label{sec:conclusion}
In this paper, a pilot design framework for LMMSE channel estimation of GFDM is formulated. In addition, we propose a set of parameters for the framework such that the inherent interference in GFDM can be precancelled during the pilot generation. To achieve this goal, we require the transmit samples at selected frequency bins to have values equal to a fixed reference sequence, and the pilot vector with its corresponding LMMSE channel estimator can be obtained. Simulation results demonstrate that the proposed method can achieve a significant reduction of channel estimation MSE and high-SNR SER, compared with the conventional methods.

\bibliographystyle{IEEEbib}
\bibliography{chiptune}
\small

\end{document}